\documentclass[conference]{IEEEtran}
\usepackage{amsmath,amssymb,amsthm,cite,graphicx,array}
\usepackage{graphicx}
\usepackage{float}
\usepackage{caption}
\usepackage{multicol}
\usepackage{mathrsfs}
\usepackage{amsmath}
\usepackage{amssymb}
\usepackage{amsthm}
\usepackage{multicol}
\usepackage{algorithm}
\usepackage[noend]{algpseudocode}
\usepackage{float}
\usepackage{placeins}
\usepackage{epstopdf}
\usepackage{multirow}
\usepackage{subfigure}
\usepackage{enumitem}
\usepackage[utf8]{inputenc}
\floatstyle{ruled}
\newfloat{algorithm}{tbp}{loa}
\providecommand{\algorithmname}{Algorithm}
\floatname{algorithm}{\protect\algorithmname}
\theoremstyle{remark}
\newtheorem{theorem}{Theorem}

\newtheorem{note}{Note}
\newtheorem{lemma}{Lemma}

\theoremstyle{remark}
\newtheorem{example}{Example}
\ifCLASSINFOpdf
\else
\fi

\title{Optimal Scalar Linear Index Codes for Symmetric and Neighboring Side-information Problems} 

\begin{document}

\author{Mahesh~Babu~Vaddi~and~B.~Sundar~Rajan\\ 
 Department of Electrical Communication Engineering, Indian Institute of Science, Bengaluru 560012, KA, India \\ E-mail:~\{vaddi,~bsrajan\}@iisc.ac.in}
 
\maketitle
\begin{abstract}
A single unicast index coding problem (SUICP) is called symmetric neighboring and consecutive (SNC) side-information problem if it has $K$ messages and $K$ receivers, the $k$th receiver $R_{k}$ wanting the $k$th message $x_{k}$ and having the side-information $D$ messages immediately after $x_k$ and $U$ ($D\geq U$) messages immediately before $x_k$. Maleki, Cadambe and Jafar obtained the capacity of this SUICP(SNC) and proposed $(U+1)$-dimensional optimal length vector linear index codes by using Vandermonde matrices. However, for a $b$-dimensional vector linear index code, the transmitter needs to wait for $b$ realizations of each message and hence the latency introduced at the transmitter is proportional to $b$. For any given single unicast index coding problem (SUICP) with the side-information graph $G$, MAIS($G$) is used to give a lowerbound on the broadcast rate of the ICP. In this paper, we derive MAIS($G$) of SUICP(SNC) with side-information graph $G$. We construct scalar linear index codes for SUICP(SNC) with length $\left \lceil \frac{K}{U+1} \right \rceil - \left \lfloor \frac{D-U}{U+1} \right \rfloor$. We derive the minrank($G$) of SUICP(SNC) with side-information graph $G$ and show that the constructed scalar linear index codes are of optimal length for SUICP(SNC) with some combinations of $K,D$ and $U$. For SUICP(SNC) with arbitrary $K,D$ and $U$, we show that the length of constructed scalar linear index codes are atmost two index code symbols per message symbol more than the broadcast rate. The given results for SUICP(SNC) are of practical importance due to its relation with topological interference management problem in wireless communication networks. 
\end{abstract}
\section{Introduction and Background}
\label{intro}
\IEEEPARstart{A}n index coding problem (ICP), comprises a transmitter that has a set of $K$ messages $\{ x_0,x_1,\ldots,x_{K-1}\}$, and a set of $m$ receivers $\{ R_0,R_1,\ldots,R_{m-1}\}$. Each receiver, $R_k=(\mathcal{K}_k,\mathcal{W}_k)$, knows a subset of messages, $\mathcal{K}_k \subset X$, called its \textit{side-information}, and demands to know another subset of messages, $\mathcal{W}_k \subseteq \mathcal{K}_k^\mathsf{c}$, called its \textit{Want-set}. The transmitter can take cognizance of the side-information of the receivers and broadcast coded messages, called the index code, over a noiseless channel. The objective is to minimize the number of coded transmissions, called the length of the index code, such that each receiver can decode its demanded message using its side-information and the coded messages.

Index coding problem with side-information is motivated by wireless broadcasting applications. In applications like video-on-demand, during the first transmission by the transmitter, each receiver may miss a part of data. The transmitter wants to transmit the missing packets to each receiver by taking the cognizance of the already received packets (side-information) at the receivers. 

The problem of index coding with side-information was introduced by Birk and Kol \cite{BiK}. An index coding problem is called unicast \cite{OnH} if the demand sets of the receivers are disjoint. An index coding problem is called single unicast if the demand sets of the receivers are disjoint and the size of each demand set is one. Any unicast index problem can be equivalently reduced to a single unicast index coding problem (SUICP) \cite{YBJK}. In a single unicast index coding problem, the number of messages equal to the number of receivers. Any SUICP with $K$ messages $\{x_0,x_1,\ldots,x_{K-1}\}$ can be expressed as a side-information graph $G$ with $K$ vertices $\{x_0,x_1,\ldots,x_{K-1}\}$. In $G$, there exists an edge from $x_i$ to $x_j$ if the receiver wanting $x_i$ knows $x_j$.

A scalar linear index code of length $N$ $(<K)$ is represented by a matrix $\mathbf{L}$ $(\in \mathbb{F}_q^{K \times N})$, where the $j$th column contains the coefficients of the $j$th coded transmission and the $k$th row $L_k$ $(\in \mathbb{F}_q^{1\times N})$ contains the coefficients used for mixing message $x_k$ in the $N$ transmissions. The broadcast vector is 
\begin{align*}
 c=xL=\sum_{k=0}^{K-1}x_kL_k.
\end{align*}

The broadcast rate \cite{ICVLP} of an index coding problem is the minimum (minimization over all mapping including nonlinear and all dimensions) number of index code symbols required to transmit such that every receiver can decode its wanted message by using the broadcasted index code symbols and its side-information. The capacity of an index coding problem is the reciprocal of the braodcast rate.

Maleki,  Cadambe and Jafar \cite{MCJ} found the capacity of SUICP(SNC) with $K$ messages and $K$ receivers, each receiver has a total of $U+D<K$ side-information, corresponding to the $U$ messages before and $D$ messages after its wanted message. In this setting, the $k$th receiver $R_{k}$ demands the message $x_{k}$ having the side-information
\begin{equation}
\label{antidote}
{\cal K}_k= \{x_{k-U},\dots,x_{k-2},x_{k-1}\}\cup\{x_{k+1}, x_{k+2},\dots,x_{k+D}\}.
\end{equation}

Let $G$ be the side-information graph of this setting. The capacity of this index coding problem setting is:
\begin{equation}
\label{capacity}
C(G)=\left\{
                \begin{array}{ll}
                  {1 ~~~~~~~~~~~~ \mbox{if} ~~ U+D=K-1}\\
                  {\frac{U+1}{K-D+U}} ~~~ \mbox{if} ~~U+D\leq K-2. 
                  \end{array}
              \right.
\end{equation}
where $U,D \in$ $\mathbb{Z},$ $0 \leq U \leq D$. That is, the broadcast rate of this index coding problem is 
\begin{align}
\label{capacity5}
\beta(G)=\frac{K-D+U}{U+1}.
\end{align}

Maleki \textit{et.al.} proved that the capacity \eqref{capacity} is achievable by $(U+1)$-dimensional vector linear index codes by using Vandermonde matrices over higher field sizes. In SUICP(SNC), if $U=0$, we refer it as one-sided side-information SUICP(SNC). In the one-sided side-information case, i.e., the cases where $U$ is zero, the $k$th receiver $R_{k}$ demands the message $x_{k}$ having the side-information,
\begin{equation}
\label{antidote1}
{\cal K}_k =\{x_{k+1}, x_{k+2},\dots,x_{k+D}\}, 
\end{equation}
\noindent
for which \eqref{capacity} reduces to
\begin{equation}
\label{capacity1}
C=\left\{
                \begin{array}{ll}
                  {1 ~~~~~~~~~~~~ \mbox{if} ~~ D=K-1}\\
                  {\frac{1}{K-D}} ~~~~~~~ \mbox{if} ~~D\leq K-2 
                  \end{array}
              \right.
\end{equation}
symbols per message.
 
Jafar \cite{TIM} established the relation between index coding problem and topological interference management problem. The SUICP(SNC) is motivated by topological interference management problem. Consider a cellular network where a user can receive and decode the signal transmitted by neighboring base stations. From an interference alignment prospective, this type of problem can be modelled as the SUICP(SNC).

Consider a single unicast index coding problem with $K$ messages and side information graph $G$. Let $$\mathbf{e}_k=(\underbrace{0~0~~\ldots~0}_{k-1}~1~\underbrace{0~0~\ldots~0}_{K-k}) \in \mathbb{F}_q^K.$$ The support of a vector $\mathbf{u} \in \mathbb{F}_q^K$ is defined to be the set supp$(\mathbf{u})=\big\{k \in [0:K-1]: u_k \neq 0\big\}$. Let $\mathbf{E} \subseteq [0:K-1]$. We denote $\mathbf{u} \lhd \mathbf{E}$ whenever supp$(\mathbf{u}) \subseteq \mathbf{E}$. Then, the $\text{minrank}_q(G$) over $\mathbb{F}_q$ is defined \cite{ECIC} as $\text{min}\{\text{rank}_{\mathbb{F}_q}(\{\mathbf{v}_k+\mathbf{e}_{k}\}_{k \in [0:K-1]}:\mathbf{v}_k \in \mathbb{F}_q^K, \mathbf{v}_k \vartriangleleft \mathcal{K}_k\}.$ For the side-information graph $G$, we define $\text{minrank}(G)$ is the minimum value of of $\text{minrank}_q(G)$ over all fields $\mathbb{F}_q$. In \cite{ECIC}, it was shown that for any given index coding problem, the length of an optimal scalar linear index code over a field $\mathbb{F}_q$ is equal to the $\text{minrank}_q(G)$ of its side-information graph. However, finding the minrank for any arbitrary side-information graph is NP-hard \cite{ECIC}.

For a graph $G$, the order of an induced acyclic sub-graph formed by removing the minimum number of vertices in $G$, is called $MAIS(G)$. In \cite{YBJK}, it was shown that $MAIS(G)$ lower bounds the broadcast rate. That is, $\beta(G) \geq MAIS(G).$

\subsection{Review of AIR matrices}
In \cite{VaR2}, for any arbitrary positive integers $m$ and $n$, we constructed binary matrices of size $m \times n (m\geq n)$ such that any $n$ adjacent rows of the matrix are linearly independent over every field. We refer these matrices as AIR matrices.

The matrix obtained by Algorithm \ref{algo2} is called the $(m,n)$ AIR matrix and it is denoted by $\mathbf{L}_{m\times n}.$ The general form of the $(m,n)$ AIR matrix is shown in   Fig. \ref{fig1}. It consists of several submatrices of different sizes as shown in Fig.\ref{fig1}.
\begin{figure*}
\centering
\includegraphics[scale=0.36]{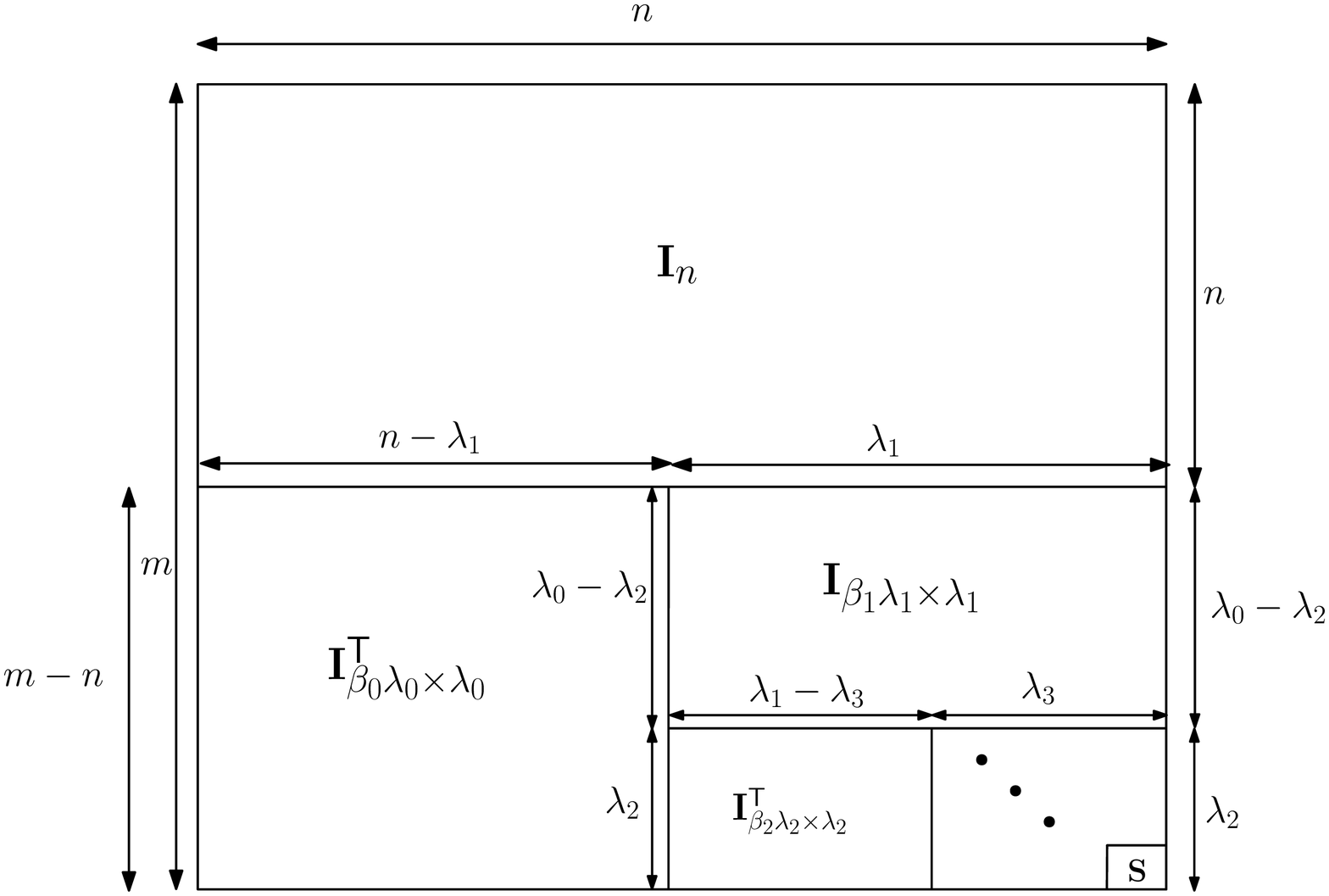}\\
~ $\mathbf{S}=\mathbf{I}_{\lambda_{l} \times \beta_l \lambda_{l}}$ if $l$ is even and ~$\mathbf{S}=\mathbf{I}_{\beta_l\lambda_{l} \times \lambda_{l}}$ otherwise.
\caption{AIR matrix of size $m \times n$.}
\label{fig1}
~ \\
\hrule
\end{figure*}
The description of the submatrices are as follows: Let $c$ and $d$ be two positive integers and $d$ divides $c$. The following matrix  denoted by $\mathbf{I}_{c \times d}$ is a rectangular  matrix.
\begin{align}
\label{rcmatrix}
\mathbf{I}_{c \times d}=\left.\left[\begin{array}{*{20}c}
   \mathbf{I}_{d}  \\
   \mathbf{I}_{d}  \\
   \vdots  \\
   \mathbf{I}_{d} 
   \end{array}\right]\right\rbrace \frac{c}{d}~\text{number~of}~ \mathbf{I}_{d}~\text{matrices}
\end{align}
and $\mathbf{I}_{d \times c}$ is the transpose of $\mathbf{I}_{c \times d}.$

Towards explaining the other quantities shown in the AIR matrix shown in Fig. \ref{fig1}, for a given $m$  and $n,$ let  $\lambda_{-1}=n,\lambda_0=m-n$ and\begin{align}
\nonumber
n&=\beta_0 \lambda_0+\lambda_1, \nonumber \\
\lambda_0&=\beta_1\lambda_1+\lambda_2, \nonumber \\
\lambda_1&=\beta_2\lambda_2+\lambda_3, \nonumber \\
\lambda_2&=\beta_3\lambda_3+\lambda_4, \nonumber \\
&~~~~~~\vdots \nonumber \\
\lambda_i&=\beta_{i+1}\lambda_{i+1}+\lambda_{i+2}, \nonumber \\ 
&~~~~~~\vdots \nonumber \\ 
\lambda_{l-1}&=\beta_l\lambda_l.
\label{chain}
\end{align}
where $\lambda_{l+1}=0$ for some integer $l,$ $\lambda_i,\beta_i$ are positive integers and $\lambda_i < \lambda_{i-1}$ for $i=1,2,\ldots,l$.

		\begin{algorithm}
		\caption{Algorithm to construct the AIR matrix $\mathbf{L}$ of size $m \times n$}
			\begin{algorithmic}[2]
				 \item Let $\mathbf{L}=m \times n$ blank unfilled matrix.
				\item [Step 1]~~~
				\begin{itemize}
				\item[\footnotesize{1.1:}] Let $m=qn+r$ for $r < n$.
				\item[\footnotesize{1.2:}] Use $\mathbf{I}_{qn \times n}$ to fill the first $qn$ rows of the unfilled part of $\mathbf{L}$.
				\item[\footnotesize{1.3:}] If $r=0$,  Go to Step 3.
				\end{itemize}

				\item [Step 2]~~~
				\begin{itemize}
				\item[\footnotesize{2.1:}] Let $n=q^{\prime}r+r^{\prime}$ for $r^{\prime} < r$.
				\item[\footnotesize{2.2:}] Use $\mathbf{I}_{q^{\prime}r \times r}^{\mathsf{T}}$ to fill the first $q^{\prime}r$ columns of the unfilled part of $\mathbf{L}$.
			    \item[\footnotesize{2.3:}] If $r^{\prime}=0$,  go to Step 3.	
				\item[\footnotesize{2.4:}] $m\leftarrow r$ and $n\leftarrow r^{\prime}$.
				\item[\footnotesize{2.5:}] Go to Step 1.
				\end{itemize}
				\item [Step 3] Exit.
		
			\end{algorithmic}
			\label{algo2}
		\end{algorithm}

In \cite{VaR1}, we constructed optimal length vector linear index codes for SUICP(SNC) by using AIR matrices. In \cite{VaR4}, we constructed optimal length scalar linear index codes for SUICP(SNC) for a special case when $U+1$ divide both $K$ and $D-U$. In \cite{VaR3}, we gave a low-complexity decoding for SUICP(SNC) with AIR matrix as encoding matrix. By using the low complexity decoding method, every receiver is able to decode its wanted message symbol by simply adding some index code symbols (broadcast symbols). 

\subsection{Contributions}

In a $b$-dimension vector index coding, the transmitter has to wait for $b$-realizations of a given message symbol to perform the index coding. In scalar index codes, the transmitter encodes each realization of the message symbols separately.  In a delay critical environment like real time video streaming, it may not be desirable to wait for $b$-realizations of a given message symbol, hence, scalar index coding is desirable.

\begin{itemize}
\item We show that $MAIS(G)$ of side-information graph $G$ of SUICP(SNC) is $\left \lfloor \frac{K-D+U}{U+1} \right \rfloor$. 
\item By using AIR matrices, we construct scalar linear index codes for SUICP(SNC) with length $\left \lceil \frac{K}{U+1} \right \rceil - \left \lfloor \frac{D-U}{U+1} \right \rfloor$. 
\item We show that the minrank of SUICP(SNC) side-information graph is $\left \lceil \frac{K}{U+1} \right \rceil - \left \lfloor \frac{D-U}{U+1} \right \rfloor$ and hence the constructed scalar linear index codes are of optimal length for some combinations of $K,D$ and $U$. 
\item We give encoding matrices for the scalar linear index codes with length $\left \lceil \frac{K}{U+1} \right \rceil - \left \lfloor \frac{D-U}{U+1} \right \rfloor$ by using AIR matrices.
\end{itemize}

The paper is organized as follows. In Section \ref{sec2}, we derive MAIS for SUICP(SNC) side-information graphs. In Section \ref{sec3}, we give a construction of scalar linear index codes for SUICP(SNC) and we prove that the constructed codes are of optimal length for certain instances of $K,D$ and $U$. We summarize the paper in Section \ref{sec4}.

\section{Maximum Acyclic Induced Subgraph Bound for SUICP(SNC) Side-Information Graphs} 
\label{sec2}
In this section, we derive the MAIS for the side-information graph $G$ of SUICP(SNC). The MAIS in Theorem \ref{thm1} is useful to quantify the gap between the length of the constructed index codes in this paper and MAIS, which is a lowerbound on the braoadcast rate. 

\begin{theorem}
\label{thm1}
Consider the side-information graph $G$ of SUICP(SNC) with $K$ messages, $D$ side-information after and $U$ side-information before the desired message. Then
\begin{align}
\label{mais}
MAIS(G)=\left \lfloor \frac{K-D+U}{U+1} \right \rfloor.
\end{align}
\end{theorem}
\begin{proof}
Let $t=\left \lfloor \frac{K-D+U}{U+1} \right \rfloor$. Consider the induced subgraph of the $t$ vertices $$\mathcal{A}=\{x_0,x_{(U+1)},x_{2(U+1)},\ldots,x_{(t-1)(U+1)}\}.$$ in $G$. Let this induced graph be $G_1$.

In $G_1$, there exists no edge from $x_{j(U+1)}$ to $x_0$ follows from the fact that $j(U+1)+D<K$ for every $j \in [1:t-1]$ as shown in \eqref{mais2}. 

\begin{align}
\label{mais2}
\nonumber
&j(U+1)+D \leq (t-1)(U+1)+D\\& \nonumber ~~~~=\left(\left \lfloor \frac{K-D+U}{U+1} \right \rfloor -1\right)(U+1)+D \\& \nonumber ~~~~=\left \lfloor \frac{K-D+U}{U+1} \right \rfloor (U+1)-(U+1)+D \\&\nonumber = \left(\frac{K-D+U}{U+1}-\frac{(K-D+U)\text{mod}(U+1)}{U+1}\right)(U+1)\\&\nonumber ~~~~~~~~~~~~~~~~~~~~~~~~~~~~~~-(U+1)+D \\& \nonumber= (K-D+U)-(K-D+U)\text{mod}(U+1)-(U+1)+D\\&= K-1-(K-D+U)~\text{mod}~(U+1) <K.
\end{align}

In $G_1$, there exists no edges from $x_{i(U+1)}$ to $x_{j(U+1)}$ if $i>j$ for every $i \in [0:t-1]$ follows from the fact that every receiver in SUICP(SNC) knows only $U$ messages before its wanted message. Hence, $G_1$ forms an acyclic induced subgraph of $G$ and 
\begin{align}
\label{mais5}
MAIS(G) \geq t.
\end{align}

But, we have 
\begin{align}
\label{mais4}
\nonumber
t+1&=\left \lfloor \frac{K-D+U}{U+1} \right \rfloor +1 \\& > \frac{K-D+U}{U+1}=\beta(G) \geq MAIS(G). 
\end{align}
Hence, from \eqref{mais2} and \eqref{mais5}, we have $MAIS(G)=t$. This completes the proof. 
\end{proof}

\begin{example}
Consider SUICP(SNC) with $K=17,D=6,U=2$. For this SUICP(SNC), from \eqref{capacity5}, $\beta(G)=\frac{13}{3}=4.25$. From Theorem \ref{thm1}, MAIS($G$) for this index coding problem is $4$ and the subgraph induced by $\{x_0,x_3,x_6,x_9\}$ forms a acyclic induced subgraph. 
\end{example}
\begin{example}
Consider SUICP(SNC) with $K=16,D=3,U=2$. For this SUICP(SNC), from \eqref{capacity5}, $\beta(G)=\frac{15}{3}=5$. From Theorem \ref{thm1}, MAIS($G$) for this index coding problem is $5$ and the subgraph induced by $\{x_0,x_3,x_6,x_9,x_{12}\}$ forms a acyclic induced subgraph. 
\end{example}
\section{Optimal Scalar Linear Index Codes for SUICP(SNC)}
\label{sec3}
In Theorem \ref{thm3}, we construct scalar linear index codes for SUICP(SNC) by using AIR matrices.
\begin{theorem}
\label{thm2}
Consider the SUICP(SNC) with $K$ messages, $D$ side-information after and $U$ side-information before the desired message. For this SUICP(SNC), there exists a scalar linear index code with length $\left \lceil \frac{K}{U+1} \right \rceil - \left \lfloor \frac{D-U}{U+1} \right \rfloor$ and the scalar linear index code can be constructed by using AIR matrix.
\end{theorem}
\begin{proof}
Let $K_1=\left \lceil \frac{K}{U+1}\right \rceil$ and $D_1=\left \lfloor \frac{D-U}{U+1} \right \rfloor$. 

Let
\begin{align}
\label{ext1}
\nonumber
&y_k=\sum_{i=0}^{U} x_{(U+1)(k-1)+i}~\text{for}~k \in [0:K_1-2]~\text{and}\\&y_{K_1-1}=\sum_{i=0}^{K~\text{mod}~(U+1)-1} x_{(U+1)(K_1-2)+i}.
\end{align}

That is, in \eqref{ext1}, we convert  $K$ message symbols into $K_1$ extended message symbols $y_k$ for every $k \in [0:K_1-1]$. Let $\mathbf{L}$ be an AIR matrix of size $K_1 \times (K_1-D_1)$. In $\mathbf{L}$, every set of $K_1-D_1$ rows are linearly independent. Let 
\begin{align}
\label{code}
 [c_0~c_1~\ldots~c_{K_1-D_1-1}]=y\mathbf{L}=\sum_{k=0}^{K_1-1}y_kL_k.
\end{align}

We show that every receiver $R_k$ decodes its wanted message symbol $x_k$ for $k \in [0:K-1]$. Let $j=\left \lceil \frac{k}{D+1} \right \rceil$ for every $k \in [0:K-1]$. We have $D_1(U+1) \leq D$ and hence, $R_k$ for $k \in [0:K_1-1]$ knows every other message symbol present in $y_{j+1},y_{j+2},\ldots,y_{j+D_1}$. 
From \eqref{code}, $R_k$ can decode $y_j$ after converting \eqref{code} into a system of $K_1-D_1$ equations with $K_1-D_1$ unknowns by removing the $D_1$ known extended message symbols $y_{j+1},y_{j+2},\ldots,y_{j+D_1}$. $R_k$ knows every other message symbol in $y_j$ except $x_k$. Hence, $R_k$ can decode $x_k$ from \eqref{code} for every $k \in [0:K-1]$. The length of the index code is $K_1-D_1=\left \lceil \frac{K}{U+1} \right \rceil - \left \lfloor \frac{D-U}{U+1} \right \rfloor$. This completes the proof.
\end{proof}
\begin{example}
\label{ex3}
Consider SUICP(SNC) with $K=20,D=9,U=2$. For this SUICP(SNC), from \eqref{capacity5}, $\beta(G)=\frac{13}{3}=4.25$. For this SUICP(SNC), we have $\underbrace{\left \lceil \frac{K}{U+1} \right \rceil}_{K_1} - \underbrace{\left \lfloor \frac{D-U}{U+1} \right \rfloor}_{D_1}=\left \lceil \frac{20}{3} \right \rceil - \left \lfloor \frac{7}{3} \right \rfloor=7-2=5$. We combine $20$ messages $x_k$ for $k \in [0:19]$ into $K_1=7$ extended messages $y_j$ for $j \in [0:6]$ as shown below.
\begin{align*}
&y_0=x_0+x_1+x_2~~~~~~~y_1=x_3+x_4+x_5\\&
y_2=x_6+x_7+x_8~~~~~~~y_3=x_9+x_{10}+x_{11}\\&
y_4=x_{12}+x_{13}+x_{14}~~~~y_5=x_{15}+x_{16}+x_{17}\\&
y_6=x_{18}+x_{19}.
\end{align*}
We use AIR matrix of size $K_1 \times (K_1-D_1)=7 \times 5$ as an encoding matrix to encode extended messages. The AIR matrix of size $7 \times 5$ is given below.

\arraycolsep=1pt
\setlength\extrarowheight{-2.0pt}
{
$$\mathbf{L}_{7 \times 5}=~\left[\begin{array}{*{20}c}
   1 & 0 & 0 & 0 & 0 \\
   0 & 1 & 0 & 0 & 0 \\
   0 & 0 & 1 & 0 & 0 \\
   0 & 0 & 0 & 1 & 0 \\
   0 & 0 & 0 & 0 & 1 \\
   1 & 0 & 1 & 0 & 1 \\
   0 & 1 & 0 & 1 & 1 \\
  \end{array}\right]$$
}

The index code for the given SUICP(SNC) is given by
\begin{align*}
\mathfrak{C}&=[c_0~c_1~c_2~c_3~c_4]=[y_0~y_1~\ldots~y_6]\mathbf{L}\\&=\{\underbrace{x_{0}+x_{1}+x_2}_{y_0}+\underbrace{x_{15}+x_{16}+x_{17}}_{y_5},\\&~~~~\underbrace{x_{3}+x_{4}+x_5}_{y_1}+\underbrace{x_{18}+x_{19}}_{y_6},\\&~~~~\underbrace{x_{6}+x_{7}+x_8}_{y_2}+\underbrace{x_{15}+x_{16}+x_{17}}_{y_5},\\&~~~~ \underbrace{x_{9}+x_{10}+x_{11}}_{y_3}+\underbrace{x_{18}+x_{19}}_{y_6},\\&~~~~ \underbrace{x_{12}+x_{13}+x_{14}}_{y_4}+\underbrace{x_{15}+x_{16}+x_{17}}_{y_5}+\underbrace{x_{18}+x_{19}}_{y_6}\}.
\end{align*}

The scalar linear code $\mathfrak{C}$ can be used to decode the message $x_k$ at $R_k$ for every $k \in [0:19]$ as given in Table \ref{table1}. For example, receivers $R_0,R_1$ and $R_2$ decode their wanted messages $x_0,x_1$ and $x_2$ respectively by adding $c_0$ and $c_2$ as shown in the first row of the Table \ref{table1}. Similarly, receivers $R_6,R_7$ and $R_8$ decode their wanted messages $x_6,x_7$ and $x_8$ respectively by adding $c_2,c_3$ and $c_4$ as shown in the third row of the Table \ref{table1}.

\begin{table}[ht]
\centering
\setlength\extrarowheight{5pt}
\begin{tabular}{|c|c|}
\hline
$x_k$ & Index code symbols used to decode $x_k$ \\
\hline
$x_0,x_1,x_2$ & $c_0,c_2$ \\
\hline
$x_3,x_4,x_5$ & $c_1,c_3$ \\
\hline
$x_6,x_7,x_8$ & $c_2,c_3,c_4$ \\
\hline
$x_9,x_{10},x_{11}$ & $c_3,c_4$ \\
\hline
$x_{12},x_{13},x_{14}$ & $c_4$ \\
\hline
$x_{15},x_{16},x_{17}$ & $c_0$ \\
\hline
$x_{18},x_{19}$ & $c_1$ \\
\hline
\end{tabular}
\vspace{5pt}
\caption{Decoding for the messages given in Example \ref{ex3}}
\label{table1}
\vspace{-5pt}
\end{table}

\end{example}
\begin{note}
The length of the constructed code in Example \ref{ex3} is equal to $\left \lceil \beta(G) \right \rceil=\left \lceil \frac{13}{3} \right \rceil=5$. In Theorem \ref{thm3}, we prove that the minrank of SUICP(SNC) given in Example \ref{ex3} is $5$ and hence the constructed scalar linear index code is of optimal length.
\end{note}

\begin{theorem}
\label{thm3}
Consider SUICP(SNC) with $K$ messages, $D$ side-information after and $U$ side-information before the desired message and side-information graph  $G$. If $\beta(G)$ is a rational number and 
\begin{align}
\label{minrank2}
\nonumber
(D-U)~\text{mod}~(U+1)&+(K-D+U)~\text{mod}~(U+1)\\&  \leq K~\text{mod}~(U+1),
\end{align}
 then 
\begin{align*}
\text{minrank}(G)=\left \lceil \frac{K}{U+1} \right \rceil - \left \lfloor \frac{D-U}{U+1} \right \rfloor=\left \lceil\beta(G) \right \rceil
\end{align*}
and the constructed scalar linear index codes in Theorem \ref{thm2} are of optimal length.
\end{theorem}
\begin{proof}
Let $\gamma=\left \lceil \frac{K}{U+1} \right \rceil - \left \lfloor \frac{D-U}{U+1} \right \rfloor$. First we prove that if $\beta(G)$ is a rational number and $K,D$ and $U$ satisfy the condition given in \eqref{minrank2}, then 
\begin{align*}
\gamma=\left \lceil\beta(G) \right \rceil.
\end{align*}

If $K~\text{mod}~(U+1)$ is zero, the condition in \eqref{minrank2} gets satisfied only when both $(D-U)~\text{mod}~(U+1)$ and $(K-D+U)~\text{mod}~(U+1)$ are zeros. We assumed $\beta(G)$ is a rational number, that is $(K-D+U)~\text{mod}~(U+1)$ is non zero and hence \eqref{minrank2} is only satisfied if $K~\text{mod}~(U+1)$ is non zero. We have
\begin{align}
\label{mrank2}
\nonumber
\gamma&=\left \lceil \frac{K}{U+1} \right \rceil - \left \lfloor \frac{D-U}{U+1} \right \rfloor\\& \nonumber=\underbrace{\frac{K}{U+1}+\frac{(U+1)-K~\text{mod}~(U+1)}{U+1}}_{\text{follows from the fact K mod (U+1) is nonzero}}\\& \nonumber +\frac{D-U}{U+1}+\frac{(D-U)~\text{mod}~(U+1)}{U+1}=\\& \nonumber\small{\frac{K-D+U}{U+1}+1+\frac{(D-U)\text{mod}(U+1)-K\text{mod}(U+1)}{U+1}}\\& \nonumber \leq \frac{K-D+U}{U+1}+1-\underbrace{\frac{(K-D+U)~\text{mod}~(U+1)}{U+1}}_{\text{follows from \eqref{minrank2}}}  \\& \nonumber =\frac{K-D+U}{U+1}+\frac{(U+1)-(K-D+U)~\text{mod}~(U+1)}{U+1}\\& = \left \lceil \frac{K-D+U}{U+1} \right \rceil=\left \lceil \beta(G) \right \rceil
\end{align} 

Notice that $\gamma$ is an integer and we constructed scalar linear index code for SUICP(SNC) with length $\gamma$ in Theorem \ref{thm2}. Hence, we have
\begin{align}
\label{mrank1}
\gamma \geq \beta(G) > \left \lceil \beta(G) \right \rceil -1.
\end{align}

From \eqref{mrank2} and \eqref{mrank1}, we have
\begin{align}
\label{mrank3}
\gamma = \left \lceil \beta(G) \right \rceil.
\end{align}

Note that minrank can not be less than $\beta(G)$ and hence from \eqref{mrank1} and the code construction given in Theorem \eqref{thm2} with length $\gamma$, we have 
\begin{align}
\label{mrank4}
\text{minrank}(G)=\gamma=\left \lceil \beta(G) \right \rceil. 
\end{align}

For the SUICP(SNC) with $K,D,U$ satisfying the condition \eqref{minrank2}, in Theorem \ref{thm2}, we constructed scalar linear index codes with length $\gamma$ and from \eqref{mrank3} and \eqref{mrank4}, the constructed codes are optimal in length. 
\end{proof}
In Theorem \ref{thm4}, we show that the length of constructed scalar linear index codes in Theorem \ref{thm2} are atmost two index code symbols per message symbol more than the broadcast rate.
\begin{theorem}
\label{thm4}
Consider SUICP(SNC) with $K$ messages, $D$ side-information after and $U$ side-information before the desired message and side-information graph $G$. Then 
\begin{align}
\label{mrank6}
\left \lceil \frac{K}{U+1} \right \rceil - \left \lfloor \frac{D-U}{U+1} \right \rfloor < \beta(G)+2,
\end{align}
that is the constructed scalar linear index codes in Theorem \ref{thm2} are atmost two index code symbols per message symbol away from the broadcast rate.
\end{theorem}
\begin{proof}
We have 
\begin{align*}
\beta(G)+2&=\frac{K-D+U}{U+1}+2 \\&=\frac{K}{U+1}+1-\frac{D-U}{U+1}+1\\&
> \left \lceil \frac{K}{D+1}\right \rceil -\left \lfloor \frac{D-U}{U+1} \right \rfloor.
\end{align*}
This completes the proof.
\end{proof}
The result of Theorem \ref{thm4} is illustrated using Example \ref{ex6} given below.
\begin{example}
\label{ex6}
Consider the set of SUICP(SNC) with $K=827,D=23$ and $U \in [1:10]$. For these index coding problems, the broadcast rate and the rate achieved by scalar linear index codes given in Theorem \ref{thm2} are given in Table \ref{table5}.
\begin{table}[ht]
\centering
\setlength\extrarowheight{5pt}
\begin{tabular}{|c|c|c|c|c|}
\hline
$K$ & $D$ & $U$ & Broadcast rate & Rate achieved by scalar linear \\
~& ~ & ~ & $\beta=\frac{K-D+U}{U+1}$ & index codes given in Theorem \ref{thm2}\\
\hline
827 & 23 & 1 & 402.5 & 403 \\
827 & 23 & 2 & 268.6 & 269 \\
827 & 23 & 3 & 201.7 & 202 \\
827 & 23 & 4 & 161.6 & 163 \\
827 & 23 & 5 & 134.8 & 135 \\
827 & 23 & 6 & 115.7 & 117 \\
827 & 23 & 7 & 101.3 & 102 \\
827 & 23 & 8 & 90.2 & 91 \\
827 & 23 & 9 & 81.3 & 82 \\
827 & 23 & 10 & 74.0 & 75 \\
\hline
\end{tabular}
\vspace{5pt}
\caption{}
\label{table5}
\vspace{-5pt}
\end{table}
\end{example}
In Lemma \ref{lemma1}, we construct encoding matrices for the index code constructed in Theorem \ref{thm2}.
\begin{lemma}
\label{lemma1}
Consider SUICP(SNC) with $K$ messages, $D$ side-information after and $U$ side-information before the desired message. For this SUICP(SNC), an encoding matrix $\mathbf{L}$ is given by 

$$\mathbf{L}^{(2s)}=\left[\begin{array}{*{20}c}
   \mathbf{L}_0  \\
   \mathbf{L}_1  \\
   \vdots  \\
   \mathbf{L}_{K_1-1}   \\
   \end{array}\right],
$$
$$
~\text{where} ~K_1=\left \lceil \frac{K}{U+1}\right \rceil,~
\mathbf{L}_k=\left.\left[\begin{array}{*{20}c}
   L_k  \\
   L_k  \\
   \vdots  \\
   L_k  \\   
   \end{array}\right]\right\rbrace (U+1)~\\ \text{rows}~\text{for}$$
   
      $$
k \in [0:K_1-2]~~\text{and}~~\mathbf{L}_{K_1-1}=\left.\left[\begin{array}{*{20}c}
   L_k  \\
   L_k  \\
   \vdots  \\
   L_k  \\   
   \end{array}\right]\right\rbrace K~\text{mod}~(U+1)~\\ \text{rows}$$
and $L_k$ is the $k$th row of AIR matrix $\mathbf{L}$ of size $K_1 \times (K_1-D_1)$ for $k \in [0:K_1-1]$. 
\end{lemma}
\begin{proof}
We have $[x_0~x_1~\ldots~x_{K-1}]\mathbf{L}^{(2s)}$ 
\begin{align*}
=&[x_0~x_1~\ldots~x_{U}]\mathbf{L}_0+[x_{U+1}~x_{U+2}~\ldots~x_{2U+1}]\mathbf{L}_1
\\&+\ldots+[x_{(K_1-2)(U+1)}~\ldots~x_{(K_1-2)(U+1)+U}]\mathbf{L}_{K_1-2}
\\&+\ldots+[x_{(K_1-1)(U+1)} \ldots x_{K-3}~x_{K-2}~x_{K-1}]\mathbf{L}_{K_1-1}
\end{align*}
\begin{align}
\label{tss}
\nonumber
=&L_0 (\sum_{i=0}^{U} x_{i})+L_1 (\sum_{i=0}^{U} x_{(U+1)+i})+\ldots\\&\nonumber+L_{K_1-2} (\sum_{i=0}^{U} x_{(U+1)(K_1-2)+i})\\& \nonumber +L_{K_1-1} (\sum_{i=0}^{K~\text{mod}~(U+1)} x_{(U+1)(K_1-2)+i}). 
\\&
=L_0 y_0+L_1 y_1+\ldots+L_{K_1-1} y_{K_1-1}=y\mathbf{L}, 
\end{align}
and it is the index code given in \eqref{code} of Theorem \ref{thm2}.
This completes the proof.
\end{proof}

Birk and Kol defined \cite{BiK} the partial clique and gave an index coding scheme for a given SUICP based on the partial cliques of the side-information graph. A directed graph $G$ with $K$ vertices is a $\kappa(G)$-partial clique iff outdeg$(x_j) \geq (K-1-\kappa)$, $\forall \ x_j \in V(G)$, and there exists a $\ x_k \in V(G)$ such that outdeg$(x_k)=(K-1-\kappa)$. For an index coding problem whose side information graph is a $\kappa(G)$-partial clique, maximum distance separable (MDS) code of length $K$ and dimension $\kappa+1$, over a finite field $\mathbb{F}_q$ for $q \geq K$, can be used as an index code. The side-information graph $G$ of SUICP(SNC) is a $(K-D-U-1)-$partial clique and MDS code of length $K$ and dimension $K-D-U$ can be used as index code. 

We proved that the minrank of SUICP(SNC) for specific combinations of $K,U$ and $D$ is $\left \lceil \frac{K}{U+1} \right \rceil - \left \lfloor \frac{D-U}{U+1} \right \rfloor$.  We conjecture that for SUICP(SNC) with arbitrary $K,D$ and $U$, minrank is minimum of $\left \lceil \frac{K}{U+1} \right \rceil - \left \lfloor \frac{D-U}{U+1} \right \rfloor$ and $K-D-U$.

\section{Discussion}
\label{sec4}
We constructed scalar linear index codes for SUICP(SNC) with arbitrary $K,D$ and $U$ that can achieve a rate within two index code symbols per message symbol from the optimal rate achieved by vector linear index codes. The constructed codes are useful in delay critical environment like real time video streaming in topological interference management problems.


\begin{thebibliography}{99}
\bibitem{BiK}
Y. Birk and T. Kol, ``Coding-on-demand by an informed-source (ISCOD) for efficient broadcast of different supplemental data to caching clients", in IEEE \textit{Trans. Inf. Theory,}, vol. 52, no.6, pp.2825-2830, June 2006.

\bibitem{OnH}
L Ong and C K Ho, ``Optimal Index Codes for a Class of Multicast Networks with Receiver side-information'', in \textit{Proc. IEEE ICC}, 2012, pp. 2213-2218.

\bibitem{YBJK}
Z. Bar-Yossef, Z. Birk, T. S. Jayram, and T. Kol, ``Index coding with side-information", in IEEE \textit{Trans. Inf. Theory,}, vol. 57, no.3, pp.1479-1494, Mar. 2011.

\bibitem{ICVLP}
A. Blasiak, R. Kleinberg and E. Lubetzky, ``Broadcasting with side-information: Bounding and approximating the broadcast rate", in IEEE \textit{Trans. Inf. Theory,}, vol. 59, no.9, pp.5811-5823, Sep. 2013.

\bibitem{MCJ}
H. Maleki, V. Cadambe, and S. Jafar, ``Index coding – an interference alignment perspective", in IEEE \textit{Trans. Inf. Theory,}, vol. 60, no.9, pp.5402-5432, Sep. 2014.

\bibitem{TIM}
S. A. Jafar, “Topological interference management through index coding,” IEEE Trans. Inf. Theory, vol. 60, no. 1, pp. 529–568, Jan. 2014.

\bibitem{VaR2}
M. B. Vaddi and B. S. Rajan, ``Optimal scalar linear index codes for one-sided neighboring side-information problems,'' \textit{In Proc. IEEE GLOBECOM Workshop on Network Coding and Applications}, Washington, USA, December 2016.

\bibitem{VaR1}
M. B. Vaddi and B. S. Rajan, ``Optimal vector linear index codes for some symmetric multiple unicast problems,'' in \textit{Proc. IEEE ISIT}, Barcelona, Spain July 2016, pp. 125-129.

\bibitem{VaR4}
M. B. Vaddi and B. S. Rajan, ``Reduced dimensional optimal vector linear index codes for index coding problems with symmetric neighboring and consecutive side-information,'' in arXiv:1801.00406v1 [cs.IT]  Jan.1,  2018.

\bibitem{VaR3}
M. B. Vaddi and B. S. Rajan, ``Low-complexity decoding for symmetric, neighboring and consecutive side-information index coding problems,'' in arXiv:1705.03192v2 [cs.IT] 16 May 2017.

\bibitem{ECIC}
S.~H. Dau, V. Skachek, and Y.~M. Chee, ``Error correction for index coding with side information,'' \textit{IEEE Trans. Inf. Theory}, vol. 59, no. 3, pp. 1517-1531, March 2013.


\end{thebibliography}
\end{document}